\documentclass[twocolumn]{revtex4-2}
\usepackage{amsmath,amsfonts,amsthm,amssymb,enumitem,color} 
\theoremstyle{remark}
\begin{document}
\newtheorem{Prop}{Proposition}
\newtheorem{Lem}{Lemma}
\title{Consequences of a dynamical no-signaling condition 
for classical-quantum interactions}
\author{S. Camalet}
\affiliation{Sorbonne Universit\'e, CNRS, Laboratoire de Physique 
Th\'eorique de la Mati\`ere Condens\'ee, LPTMC, F-75005, Paris, France}
\begin{abstract}
Hybrid classical-quantum approaches are instrumental in numerous fields, 
from condensed matter physics to quantum information science. We recently 
proposed to describe hybrid systems starting from a set of natural axioms for 
measurement probabilities without adding any underlying mathematical 
structure. The so-defined probability measures fulfill a no-signaling condition 
that ensures that instantaneous communication is impossible. We formulate 
here a dynamical generalization of this condition. It means that, for two 
independent systems, the outcome probabilities of a measurement made on 
one of them are not affected by a measurement performed earlier on the other. 
Analogous requirements are satisfied for usual classical and quantum bipartite 
systems and violating them would make faster-than-light signaling possible. 
The dynamical no-signaling condition has important consequences for 
classical-quantum interactions that depend on the hybrid approach 
considered, characterized by additional assumptions. For no-signaling hybrid 
dynamics with classical trajectories, the classical degrees of freedom can 
influence the quantum ones but the latter cannot react on the former. If the 
pure states of quantum systems remain pure then the dynamical no-signaling 
condition implies the absence of classical reaction. When all hybrid states are 
allowed, the classical-quantum interactions are strongly restricted for 
no-signaling hybrid dynamics that do not generate correlations between the 
classical and the quantum degrees of freedom. In all these cases, the proposed 
condition is equivalent to the convex-linearity of the probability measure 
transformations describing finite-time evolutions.
\end{abstract} 
\maketitle
\section{Introduction}
Classical-quantum systems are of relevance in numerous fields. Such hybrid 
approaches are, for instance, considered to describe quantum measurements 
in the Copenhagen interpretation \cite{Bohr,dE} and to study the 
spacetime-matter interaction in the absence of a quantum theory of gravity 
\cite{BT,An,O}. They are also used in quantum chemistry and condensed 
matter physics when a full quantum treatment is computationally infeasible 
\cite{KC,CB} and in quantum information science to define important quantum 
resources such as entanglement \cite{HHHH,CG} and to design useful 
algorithms \cite{CC}. 

Very different approaches have been developed and basic requirements that 
a hybrid description should meet have been proposed to apprehend this wide 
variety of approaches \cite{BT,T1}. They are that the classical degrees of 
freedom are characterized by a positive probability density function, the 
quantum ones are characterized by a positive density operator, the usual 
classical and quantum dynamics are recovered in the absence of classical-
quantum interactions, superluminal communication is impossible and the 
quantum degrees of freedom cannot remain in a pure state when genuine 
classical-quantum interactions are present.

Hybrid dynamical approaches usually start by assuming a mathematical form 
for the systems states \cite{BCGG,T1}. Some descriptions involve usual classical 
trajectories and quantum states \cite{BT,An,CB}. In other approaches, both 
classical and quantum degrees of freedom are described within the same formal 
framework either classical or quantum \cite{E,ABCMP,A,Ge,Koo,S,BJ,PT,T2,D}. 
One or several of the above-listed requirements are violated by most of these 
hybrid schemes \cite{T1,BCGG,PT,T2}. In a recent paper, we proposed a 
different starting point which is a set of natural axioms for measurement 
probabilities \cite{PRA}. These axioms ensure a positive probability measure for 
the classical degrees of freedom and a positive density operator for the quantum 
ones. Moreover, they fully determine hybrid states without assuming any 
additional underlying mathematical structure. 

Both classical and quantum joint measurement probabilities are no-signaling 
\cite{Be}. For any two systems and measurements performed on them, the 
outcome probabilities for one system do not depend on the measurement 
made on the other. No instantaneous signaling between distant systems is 
possible by performing measurements on them. The no-signaling condition can 
be stated for classical, quantum and hybrid systems as it is expressed in terms 
of the measurement probability measure of the two systems considered. As we 
will see, it follows straightforwardly from the above-mentioned axioms for 
hybrid probability measures.

In this paper, we formulate a requirement for the finite-time evolutions of hybrid 
probability measures which is weaker than that previously considered \cite{PRA} 
but yet has important consequences for classical-quantum interactions. We 
generalize the no-signaling condition taking into account independent 
time evolutions, i.e., such that correlations between distinct systems do not 
influence the evolution of the marginal probability measure of any of them. To 
do so, we consider the following protocol. A pre-measurement is performed on 
an ancillary system, the systems then evolve independently of one another and 
finally a measurement is made on the system of interest. We assume that the 
probability of obtaining any given outcome in the last measurement does not 
depend on the chosen pre-measurement. This assumption leads to a condition 
for the probability measure transformations describing finite-time evolutions. 
A transformation violating this requirement could be used for faster-than-light 
communication between spacelike-separated pre-measurement and 
measurement. The consequences of this dynamical no-signaling condition for 
classical-quantum interactions depend on the hybrid approach used. Several 
cases are examined in this paper.

As we will see, an analogous condition is fulfilled for two non-interacting 
quantum systems. Let us discuss this case more explicitly. Assume that the two 
systems, labeled a and b, are in a common correlated state and that a 
dichotomic pre-measurement is performed on a. Let $\rho_n$ be the state of b 
immediately after the pre-measurement when outcome $n$ occurs. Any 
measurement probability for b alone is fully determined by the density operator 
$\rho=p \rho_1 + (1-p) \rho_2$, where $p$ is the probability of outcome $1$ 
in the pre-measurement and $\rho$ is the state of b obtained simply by tracing 
out a. Whereas $p$, $\rho_1$ and $\rho_2$ depend on the pre-measurement, 
$\rho$ does not. Suppose now that the measurement on b is made after a 
delay. Under the assumption that a and b do not interact, the corresponding 
probabilities follow from $\rho'=p \rho'_1 + (1-p) \rho'_2$, where $\rho'$ is 
obtained by time-evolving $\rho$ and $\rho'_n$ by time-evolving $\rho_n$. In 
this paper, we assume that independent hybrid systems satisfy a similar 
condition.

The rest of the paper is organized as follows. In Sec.~\ref{Hscrm}, we recall the 
above-mentioned axioms and the ensuing hybrid states. The particular hybrid 
descriptions with classical probability density functions are also discussed. In 
Sec.~\ref{Ahs}, we formulate no-signaling conditions for the finite-time 
evolutions of a hybrid system in the presence of an ancillary system that is 
classical, quantum or hybrid. It is also shown that analogous no-signaling 
conditions are satisfied for usual classical and quantum systems. In Sec.~\ref{St}, 
we examine examples for which such dynamical no-signaling requirements are 
violated. In Sec.~\ref{Eq}, we prove that, when all hybrid probability measures 
are allowed, the conditions introduced in Sec.~\ref{Ahs} are equivalent to each 
other and to the convex-linearity of the probability measure transformations 
describing finite-time evolutions. This result also holds for approaches with 
classical probability density functions. In the following sections, we study 
specific no-signaling hybrid dynamics defined by additional assumptions. In 
Sec.~\ref{Nqb}, it is shown that, in the particular case of no-signaling 
dynamics with classical trajectories, the classical degrees of freedom, whose 
probability measures are restricted in this case to Dirac measures, are not 
influenced by the quantum ones. In Sec.~\ref{Ncb}, the hybrid dynamics is 
posited to be no-signaling and to preserve the purity of the states of quantum 
systems. Under these conditions, we find that there is no classical reaction. In 
Sec.~\ref{Ngcqi}, we determine the possible classical-quantum interactions 
assuming a dynamics that is no-signaling and generates no correlations 
between the classical and the quantum degrees of freedom. Finally, in 
Sec.~\ref{C}, we conclude and summarize our main results.
\section{Hybrid probability measures}\label{Hscrm}
We are interested in hybrid systems, denoted by $h$, consisting of classical 
and quantum degrees of freedom. The classical subsystem of $h$, denoted by 
$c_h$, is characterized by a sample space $X$, which can be any set, e.g., 
a phase space, and an event space $\cal A$, which is a $\sigma$-algebra on 
$X$. The quantum subsystem, denoted by $q_h$, is characterized by a 
separable Hilbert space ${\cal H}$ and the set $\cal E$ of positive operators 
$E$ on ${\cal H}$ such that $I-E$ is also positive, where $I$ is the identity 
operator on ${\cal H}$. The elements of $\cal E$ are termed effects and the 
weak topology is considered on this set \cite{K}. The elements of $X$ are 
denoted by $x$. 

A probability measure $w$ of the hybrid system $h$ is a map from 
$\cal A \times \cal E$ to $\mathbb{R}^+$ such that 
$w(\cup_n A_n,E)=\sum_n w(A_n,E)$ for any sequence of pairwise disjoint 
events $A_n$ and any effect $E$, $w(A,\sum_n E_n)=\sum_n w(A,E_n)$ 
for any sequence of effects $E_n$ such that $\sum_n E_n$ is an effect and 
any event $A$ and $w(X,I)=1$ \cite{PRA}. The probability measure of $c_h$ 
is $p : A \mapsto w(A,I)$ and that of $q_h$ is the map $E \mapsto w(X,E)$. 
It satisfies the prerequisites of the Busch-Gleason theorem and so there is a 
density operator $\rho$ such that $w(X,E)=\mathrm{tr}(\rho E)$ \cite{G,B}. 
Note that the set of all probability measures of $h$ is convex. 

Any hybrid probability measure $w$ can be written as
\begin{equation}
w : (A,E) \mapsto \int_A \mathrm{tr}(\eta(x) E) dp(x) , \label{Gt}
\end{equation}
where $p$ is a probability measure on $\cal A$ and $\eta$ is a $p$-integrable 
map from $X$ to the set of density operators on $\cal H$, i.e., such that 
$x \mapsto \mathrm{tr}(\eta(x) M)$ is $p$-integrable for any bounded operator 
$M$ on $\cal H$ \cite{PRA}. For given $w$, $p$ is unique and $\eta$ is 
$p$-almost everywhere unique. The probability measure of $c_h$ is $p$ and 
the density operator $\rho$ of $q_h$ is equal to the Bochner integral 
$\int \eta dp$, where, as usual, the domain of integration is $X$ when it is 
omitted \cite{BI}. Any measurement probability $w(A,E)$ can be evaluated 
using Eq.\eqref{Gt}. Thus, the classical probability measure $p$ and the map 
$\eta$ can be considered, together, as the state of the hybrid system $h$. 
Note that, in the particular case of a discrete classical subsystem, i.e., 
when the sample space $X$ is a subset of $\mathbb{N}$ and $\cal A$ is the 
set of all subsets of $X$, Eq.\eqref{Gt} can be written as 
$w : (A,E) \mapsto  \mathrm{tr}(\Omega_{p,\eta} {\hat A}\otimes E)$. 
Here, the operator ${\hat A}$ is defined as 
${\hat A}=\sum_{x \in A} |x \rangle \langle x |$, 
where $\{| x \rangle\}_{x \in X}$ is an orthonormal basis of a Hilbert space 
whose dimension equals the cardinality of $X$, while $\Omega_{p,\eta} 
= \sum_{x \in X} p(\{ x \}) |x \rangle \langle x | \otimes \eta(x)$ is a familiar 
classical-quantum state \cite{HHHH}.

Most descriptions of hybrid systems do not take into account all the hybrid 
probability measures defined above but only some of them that fulfill additional 
requirements. For instance, it is usual to assume that classical probability 
measures are characterized by probability density functions with respect to a 
natural reference measure $\mu$, e.g., the counting measure if $X$ is discrete 
or the Lebesgue measure if $X$ is a Euclidean space \cite{T1}. Moreover, 
$\mu$ is usually $\sigma$-finite, i.e., $X$ is a countable union of events with 
finite measure. So, due to the Radon-Nikodym theorem, the above requirement 
is equivalent to the vanishing of the classical probability $p(A)$ for any 
$\mu$-null event $A$ \cite{Bi}. 

Hybrid probability measures $w$ for which $w(A,I)=0$ for any $\mu$-null 
event $A$, with $\mu$ a $\sigma$-finite measure, can be written as 
$w : (A,E) \mapsto \int_A \mathrm{tr}(\omega(x) E) d\mu(x)$, where 
$\omega$ is a $\mu$-integrable map from $X$ to the set of positive 
trace-class operators on $\cal H$ such that 
$\mathrm{tr} \int  \omega d\mu=1$ \cite{PRA}. For given $w$, $\omega$ 
is $\mu$-almost everywhere unique. The probability measure $p$ of $c_h$ 
is characterized by the probability density function 
$f : x \mapsto \mathrm{tr}\omega(x)$. The state of $q_h$ is 
$\rho=\int \omega d\mu$ and the quantum state for given $x$, such that 
$f(x)>0$, is $\eta(x)=\omega(x)/f(x)$. Any measurement probability $w(A,E)$ 
can be evaluated using the above expression. Thus, the map $\omega$ can be 
considered as the state of the hybrid system $h$ in this case. Such states are 
used in many hybrid approaches 
\cite{O,KC,T1,ABCMP,A,D,ABMCCGJG,OSSW,LOW}. 
Note that the corresponding sets of hybrid probability measures and of classical 
probability measures are convex.  
In the particular case of a discrete classical subsystem, $\mu$ is the 
counting measure and $f : x \mapsto p(\{ x \})$. Hence, with the notation 
introduced above, $\Omega_{p,\eta}$ reduces to 
$\sum_{x \in X} |x \rangle \langle x | \otimes \omega(x)$.
\section{No-signaling conditions for finite-time evolutions}\label{Ahs}
We assume that, for distinct systems, there exist independent time evolutions 
defined as follows. Consider a hybrid system $h$ of interest, characterized by 
a set $\cal W$ of hybrid probability measures, and any ancillary system. When 
they evolve independently from each other, any finite-time evolution of $h$ is 
described by a map ${\cal T}$ from $\cal W$ to itself. In other words, possible 
initial correlations between $h$ and the ancilla do not influence how the state 
of $h$ alone evolves. We suppose, moreover, that the transformation $\cal T$ 
fulfills the following no-signaling condition. The outcome probabilities of any 
measurement performed on $h$ after a pre-measurement made on the 
ancilla and the finite-time evolution described by $\cal T$ do not depend on 
the chosen pre-measurement. A pre-measurement on the ancillary system 
results in an ensemble of probability measures for $h$. While such an ensemble 
depends on the pre-measurement, the unconditional probability measure for 
$h$ does not, as we will see. So, different ensembles cannot be distinguished 
by performing measurements on $h$. The dynamical no-signaling condition 
considered in this paper means that this indistinguishability persists under an 
independent time evolution. The cases of an ancilla consisting of a 
quantum system, a classical system, or a hybrid system are examined below.

Let $hq$ be the bipartite system made up of $h$ and any quantum system 
$q$ with Hilbert space ${\cal H}_q$. It is a hybrid system whose events are 
those of $h$ and effects are the positive operators $G$ on 
${\cal H} \otimes {\cal H}_q$ such that $I\otimes I_q-G$ is also positive, 
where $I_q$ is the identity operator on ${\cal H}_q$. Let $w$ be any 
probability measure of $hq$. The corresponding probability measure of $q$ is 
$w_q : F \mapsto w(X,I\otimes F)$. For any effect $F$ of $q$ such that 
$w_q(F)$ is nonvanishing, a map $w_F$ on $\cal A \times \cal E$ can be 
defined by
\begin{equation}
w_F(A,E)=\frac{w(A,E\otimes F)}{w(X,I\otimes F)} . \label{wF}
\end{equation}
It is a probability measure, as it satisfies $w_F(X,I)=1$ by construction and any 
of the other conditions in the probability measure definition as soon as $w$ 
does. The probabilities given by Eq.\eqref{wF} can be interpreted as 
conditional probabilities given the outcome, characterized by the effect $F$, of 
a measurement performed on $q$. With this notation, the probability measure 
of $h$ is $w_{I_q} : (A,E) \mapsto w(A,E\otimes I_q)$. For any finite 
sequence of effects $F_n$ such that $\sum_n F_n = I_q$, which can 
correspond to the outcomes $n$ of a measurement made on $q$, with 
probabilities $w_q(F_n)$, one has $w_{I_q}= \sum_n w_q(F_n) w_{F_n}$ 
where $n$ runs over all integers such that $w_q(F_n)>0$. This no-signaling 
condition means that the probability of any outcome, characterized by the 
event $A$ and the effect $E$, of a measurement performed on $h$ after a 
pre-measurement made on $q$ does not depend on this pre-measurement 
and is always given by $w_{I_q}(A,E)$. 

As mentioned above, we assume that any transformation $\cal T$ of $h$ 
fulfills the no-signaling condition
\begin{equation}
{\cal T}(w_{I_q})= \sum_n w_q(F_n) {\cal T}(w_{F_n}) , \label{nsc}
\end{equation}
where $n$ runs over all integers such that $w_q(F_n)>0$. The right-hand 
side of this equality corresponds to a protocol in which a pre-measurement, 
characterized by the effects $F_n$, made on $q$ is followed by an 
independent finite-time evolution of $hq$ which is described by ${\cal T}$ for 
$h$. The probability of obtaining the outcome $n$ in the pre-measurement, 
and the event $A$ and the effect $E$ in a measurement performed on $h$ 
after the finite-time evolution, is $w_q(F_n) {\cal T}(w_{F_n})(A,E)$. So, the 
probability of obtaining $A$ and $E$ regardless of the pre-measurement 
outcome is given by the right-hand side of Eq.\eqref{nsc}. Condition \eqref{nsc} 
means that, for any pre-measurement, this probability is equal to 
${\cal T}(w_{I_q})(A,E)$ which corresponds to the particular case in which no 
pre-measurement is performed on $q$. 

Note that this condition is weaker than the complete positivity assumption 
considered in Ref.~\cite{PRA}. The existence of a probability measure $w'$ of 
$hq$ such that $w'(A,E\otimes F)=w_q(F) {\cal T}(w_{F})(A,E)$ is not assumed 
here. The difference between the two assumptions can be made explicit in 
the specific case of a discrete classical subsystem. Here, a hybrid probability 
measure $w$ determines a unique hybrid state $\omega$ such that 
$w : (A,E) \mapsto \sum_{x \in A} \mathrm{tr}(\omega(x)E)$, see 
Sec.~\ref{Hscrm}. Any convex-linear map $T$ on the set of hybrid states fulfills 
the dynamical no-signaling condition considered in this paper, as shown in 
Sec.~\ref{Eq}. But the only maps $T$ that meet the complete positivity 
assumption are those such that
\begin{equation}
T(\omega) : x \mapsto \sum_{x' \in X} \sum_{\alpha} 
L_{\alpha}(x,x') \omega(x') L_{\alpha}(x,x')^\dag ,
\nonumber
\end{equation}
where the sequence of maps $L_{\alpha}$ from $X^2$ to the set of bounded
operators on ${\cal H}$ satisfies, for any $x' \in X$, 
$\sum_{x \in X,\alpha} L_{\alpha}(x,x')^\dag L_{\alpha}(x,x')=I$ \cite{PRA}.

Let us discuss the case of two non-interacting quantum systems. Any 
finite-time evolution of such a bipartite system is described by a quantum 
operation $O_{bs}$ with product Kraus operators 
$K_{\alpha} \otimes L_{\beta}$ \cite{K}. Let $\Gamma$ denote the initial 
quantum state of the bipartite system. The corresponding probability measure 
of the second system, named $q$, is 
$w_q : F \mapsto \mathrm{tr} (\Gamma I \otimes F)$. The state of the 
bipartite system after obtaining the outcome $n$ in a measurement, 
characterized by the operators $M_n$, made on $q$, is $\Gamma_n = 
I \otimes M^{\phantom{\dag}}_n \Gamma I \otimes M^{\dag}_n/w_q(F_n)$, 
where $F_n=M^{\dag}_n M^{\phantom{\dag}}_n$. The corresponding state 
of the first system is $\rho_{F_n}=\mathrm{tr}_q \Gamma_n
=\mathrm{tr}_q (\Gamma I \otimes F_n)/w_q(F_n)$, where $\mathrm{tr}_q$ 
denotes the partial trace over ${\cal H}_q$, which leads to the analog of 
Eq.\eqref{wF}. After the finite-time evolution described by $O_{bs}$, it is 
$\mathrm{tr}_q O_{bs}(\Gamma_n)=O(\rho_{F_n})$, where $O$ is the first 
system quantum operation with Kraus operators $K_{\alpha}$. In other words, 
the two systems evolve independently from each other as defined above. 
Moreover, the quantum operation $O$ satisfies 
$O(\rho_{I_q})=\sum_n w_q(F_n) O(\rho_{F_n})$, and so the analog of 
Eq.\eqref{nsc} is fulfilled. This equality has already been discussed as 
a consequence of the no-signaling principle, rather than of the convex-linearity 
of $O$, in the following way \cite{SBG}. The ensemble 
$\{ (p_n, \rho_{F_n}) \}_n$, where $p_n=w_q(F_n)$, evolves into the 
ensemble $\{ (p_n, O(\rho_{F_n})) \}_n$, which gives the density operator 
$\sum_n p_n O(\rho_{F_n})$. Since different ensembles generated by the 
pre-measurement cannot be distinguished at any subsequent time, this 
operator equals $O(\rho_{I_q})$.

Consider now the bipartite system $hc$ consisting of the hybrid system $h$ 
and any classical system $c$. Let $c$ be characterized by the sample space 
$Y$ and the event space ${\cal B}$. The bipartite system $hc$ is a hybrid 
system whose events are the elements of the product $\sigma$-algebra 
${\cal A} \otimes {\cal B}$ and effects are the elements of $\cal E$. Let $w$ 
be any probability measure of $hc$. The corresponding probability measure of 
$c$ is $p_c : B \mapsto w(X \times B,I)$. For any event $B$ of $c$ such that 
$p_c(B)>0$, a probability measure $w_{B}$ of $h$ can be defined by 
$w_B(A,E)=w(A \times B,E)/p_c(B)$. It can be interpreted in a similar way to 
$w_F$ given by Eq.\eqref{wF}. In analogy with Eq.\eqref{nsc}, we assume 
that any transformation $\cal T$ of $h$ fulfills, for any finite sequence of 
pairwise disjoint events $B_n$ such that $\cup_n B_n = Y$, 
\begin{equation}
{\cal T}(w_{Y})= \sum_n p_c(B_n) {\cal T}(w_{B_n}) , \label{nsc2}
\end{equation}
where $n$ runs over all integers such that $p_c(B_n)>0$. Note that this equality 
is trivially satisfied when $\cal T$ is the identity transformation $w \mapsto w$.

Let us discuss the case of two non-interacting classical systems with respective 
reference measures $\mu$ and $\nu$ and whose evolution is stochastic. 
Let $\Gamma$ denote their initial probability density function. As the two 
systems do not interact with each other, it reads, after a finite-time evolution,
$$T_{bs}(\Gamma) \!:\! (x,y)\! \mapsto  \!\!\!
\int_{X \times Y} \! h_1(x,x') h_2(y,y') \Gamma(x',y') d\xi(x',y') , $$
where $h_1$ ($h_2$) is a non-negative measurable function on $X \times X$ 
($Y \times Y$) such that its integral with respect to its first argument on $X$ 
($Y$), for any value of its second argument, equals unity, and $\xi$ is the 
product measure $\mu\otimes\nu$. The initial probability measure of the 
second system, named $c$, is 
$p_c : B \mapsto \int_{X \times B} \Gamma(x,y) d\xi(x,y)$. The probability 
density function of the bipartite system after obtaining the outcome $n$ in a 
measurement, characterized by the events $B_n$, made on $c$, is 
$\Gamma_n : (x,y) \mapsto \Gamma(x,y) 1_{B_n}(y)/p_c(B_n)$, where 
$1_B$ denotes the indicator function of the set $B$. The corresponding 
marginal probability density function of the first system is 
$f_{B_n} : x \mapsto \int_{B_n} \Gamma(x,y) d\nu(y) /p_c(B_n)$. After 
the finite-time evolution described by $T_{bs}$, it is 
$x \mapsto \int_Y T_{bs}(\Gamma_n)(x,y) d\nu(y)$, which is equal to 
$T(f_{B_n})$, with $T$ the first system transformation that changes any 
probability density function $f$ into $x \mapsto \int_X h_1(x,x') f(x') d\mu(x')$, 
i.e., the two systems evolve independently from each other. Moreover, the 
transformation $T$ obeys $T(f_Y)=\sum_n p_c(B_n) T(f_{B_n})$, and so the 
analog of Eq.\eqref{nsc2} is fulfilled. 

Assume now that the systems evolve deterministically according to 
$T_{bs}(\Gamma) : (x,y) \mapsto 
\Gamma(\phi_1^{-1}(x),\phi_2^{-1}(y))$, where $\phi_1$ and $\phi_2$ are 
invertible measure-preserving maps, i.e., such that 
$\mu(\phi_1^{-1}(A))=\mu(A)$ and $\nu(\phi_2^{-1}(B))=\nu(B)$ for any 
events $A$ and $B$ \cite{LE}. The probability measure $p_c$ and probability 
density functions $\Gamma_n$ and $f_{B_n}$ are given by the above 
expressions. The marginal probability density function of the first system after 
obtaining the outcome $n$ in the measurement performed on $c$ 
and the finite-time evolution described by $T_{bs}$ is equal to $T(f_{B_n})$, 
where $T$ is in this case the first system transformation that changes any 
probability density function $f$ into $x \mapsto f(\phi_1^{-1}(x))$. So, the two 
systems evolve independently from each other, and the analog of 
Eq.\eqref{nsc2} is fulfilled.

Conditions \eqref{nsc} and \eqref{nsc2} can be generalized as
\begin{equation}
{\cal T}(w_{Y,I_q})= \sum_{m,n} w_{h'}(B_m,F_n) 
{\cal T}(w_{B_m,F_n}) , \label{nsc3}
\end{equation}
where $w$ is any probability measure of a bipartite system consisting of $h$ 
and any other hybrid system $h'$ made up of a classical subsystem $c$ and 
a quantum subsystem $q$. The pairwise disjoint events $B_m$ are such that 
$\cup_m B_m=Y$, the effects $F_n$ are such that $\sum_n F_n=I_q$, 
$w_{h'} : (B,F) \mapsto w(X\times B,I \otimes F)$ denotes the probability 
measure of $h'$, $m$ and $n$ run over all integers such that 
$w_{h'}(B_m,F_n)>0$, and 
$w_{B,F} : (A,E) \mapsto w(A \times B,E \otimes F)/w_{h'}(B,F)$. 
The probability measures of $h$, $hc$ and $hq$ are, respectively, 
$w_{Y,I_q}$, $(C,E) \mapsto w(C,E \otimes I_q)$, and 
$(A,G) \mapsto w(A \times Y,G)$. Condition \eqref{nsc3} reduces to 
Eq.\eqref{nsc} when the sequence $( B_m)_m$ consists only of $Y$ and to 
Eq.\eqref{nsc2} when the sequence $( F_n)_n$ consists only of $I_q$. 
\section{Signaling transformations}\label{St}
We are interested here in bipartite systems made up of a system $s$ and 
an ancillary system that evolve independently from each other, as defined 
above, but with a transformation for $s$ that does not fulfill the corresponding 
no-signaling condition analogous to Eqs.\eqref{nsc}--\eqref{nsc3}. This is 
only possible when this transformation is non-linear, see the following section. 
As the dynamical no-signaling condition is violated, a pre-measurement made 
on the ancilla can be detected by performing a measurement on $s$. In 
the ideal cases discussed below, one bit of information is transmitted with 
certainty from one system to the other. If the completion of the 
pre-measurement on the ancilla and the start of the measurement on $s$ 
are spacelike separated, this transmission is superluminal.

Consider first quantum systems and the transformation $T$ that changes any 
state $\rho$ of $s$ into 
$r(\mathrm{tr}(\rho^2)) E_1+(1-r(\mathrm{tr}(\rho^2))) E_2$ where $E_1$ 
and $E_2$ denote one-dimensional orthogonal projectors such that 
$E_1E_2=0$ and $r$ is a map from $[0,1]$ to itself such that $r(1/2)=1$ and 
$r(1)=0$. Suppose that $s$ and the ancilla, named $q$, are initially in the 
state $(E_1 \otimes F_1+E_2 \otimes F_2)/2$ where $F_1$ and $F_2$ denote 
one-dimensional orthogonal projectors of $q$ such that $F_1F_2=0$. If no 
pre-measurement is made on $q$, the state of $s$ is 
$\rho_{I_q}=(E_1+E_2)/2$ and $T(\rho_{I_q})=E_1$. If the 
pre-measurement characterized by the effects $F_1$, $F_2$ and 
$F_3=I-F_1-F_2$ is performed on $q$, outcomes 1 and 2 each occur with 
probability $1/2$ and the resulting state of $s$ is $\rho_{F_n}=E_n$ where 
$n$ denotes the outcome. In all cases, it is changed by $T$ into $E_2$, and 
so $T(\rho_{I_q}) \neq (T(\rho_{F_1})+T(\rho_{F_2}))/2$. Moreover, 
performing the measurement characterized by the effects $E_1$, $E_2$ and 
$E_3=I-E_1-E_2$ on $s$ after the transformation $T$ allows one to know 
with certainty whether or not the pre-measurement on $q$ has been made.
For quantum systems, the relation between non-linear dynamics and 
superluminal communication is well known \cite{SBG,Gi}.

For classical systems with respective reference measures $\mu$ and $\nu$, 
let $T$ be the transformation that changes any probability density function $f$ 
of $s$ into $r(p_1^2+p_2^2)f_1+(1-r(p_1^2+p_2^2))f_2$ where 
$f_n=1_{A_n}/\mu(A_n)$ and $p_n=\int_{A_n} f d\mu$ with the disjoint 
events $A_1$ and $A_2$. Suppose that the initial probability density function of 
$s$ and the ancilla, named $c$, is $(f_1 g_1+f_2 g_2)/2$ where 
$g_n=1_{B_n}/\nu(B_n)$ with the disjoint events $B_1$ and $B_2$ of $c$. If 
no pre-measurement is made on $c$, the probability density function of $s$ is 
$f_Y=(f_1+f_2)/2$ and $T(f_Y)=f_1$. If the pre-measurement characterized 
by the events $B_1$, $B_2$ and $B_3=Y \setminus (B_1\cup B_2)$ is performed 
on $c$, outcomes 1 and 2 each occur with probability $1/2$ and the resulting 
probability density function of $s$ is $f_{B_n}=f_n$ where $n$ denotes the 
outcome. In all cases, it is changed by $T$ into $f_2$, and so 
$T(f_Y) \neq (T(f_{B_1})+T(f_{B_2}))/2$. Moreover, performing the 
measurement characterized by the events $A_1$, $A_2$ and 
$A_3=X \setminus (A_1\cup A_2)$ on $s$ after the transformation $T$ allows 
one to know with certainty whether or not the pre-measurement on $c$ has 
been made.

For hybrid systems $h$ and $h'$ with respective reference measures $\mu$ 
and $\nu$, consider, for instance, the transformation $T$ that changes any 
hybrid state of $h$ into $x \mapsto 
(r(\mathrm{tr}(\rho^2))f_1(x)+(1-r(\mathrm{tr}(\rho^2)))f_2(x))\rho$, where 
$\rho$ is the state of $q_h$ that remains unchanged, and $f_1$ and $f_2$ are 
defined above. Suppose that the initial state of the bipartite system is 
$(x,y) \mapsto f(x)g(y) (E_1 \otimes F_1+E_2 \otimes F_2)/2$, with the 
probability density functions $f$ and $g$ for $c_h$ and the classical subsystem 
of $h'$, respectively, and the one-dimensional orthogonal projectors $E_1$, 
$E_2$, $F_1$ and $F_2$ such that $E_1E_2=F_1F_2=0$. If no pre-measurement 
is made on $h'$, whose quantum subsystem is $q$, the state of $h$ is 
$\omega_{Y,I_q} : x \mapsto f(x)(E_1+E_2)/2$ and 
$T(\omega_{Y,I_q}) : x \mapsto f_1(x)(E_1+E_2)/2$. If the pre-measurement 
characterized by the effects $F_1$, $F_2$ and $F_3=I_q-F_1-F_2$ is performed 
on $q$, outcomes 1 and 2 each occur with probability $1/2$ and the resulting 
state of $h$ is $\omega_{Y,F_n} : x \mapsto f(x)E_n$, where $n$ denotes the 
outcome. One has $T(\omega_{Y,F_n}) : x \mapsto f_2(x)E_n$, and so 
$T(\omega_{Y,I_q}) \neq (T(\omega_{Y,F_1})+T(\omega_{Y,F_2}))/2$. 
Moreover, performing the measurement characterized by the events $A_1$, 
$A_2$ and $A_3$ on the classical system $c_h$ after the transformation $T$ 
allows one to know with certainty whether or not the pre-measurement on the 
quantum system $q$ has been made.
\section{Equivalence of the no-signaling conditions and convex-linearity}
\label{Eq}
If a transformation $\cal T$ on a convex set $\cal W$ of hybrid probability 
measures is convex-linear, i.e., 
${\cal T}(tw_1+(1-t)w_2)=t{\cal T}(w_1)+(1-t){\cal T}(w_2)$ for any 
elements $w_1$ and $w_2$ of $\cal W$ and $t \in [0,1]$, it fulfills the 
no-signaling conditions \eqref{nsc}--\eqref{nsc3} since these equalities are 
satisfied when $\cal T$ is the identity transformation. Conversely, each of 
these no-signaling conditions implies the convex-linearity of $\cal T$ when all 
hybrid probability measures are allowed, as shown by the proposition below.
\begin{Prop}\label{Clr}
Let $\cal T$ be any map from the set of probability measures of a hybrid 
system to itself. The following assertions are equivalent. 
\begin{enumerate}[label=(\roman*)]
\item $\cal T$ fulfills the no-signaling condition \eqref{nsc}, \label{i}
\item $\cal T$ fulfills the no-signaling condition \eqref{nsc2}, \label{ii}
\item $\cal T$ fulfills the no-signaling condition \eqref{nsc3}, \label{iii}
\item $\cal T$ is convex-linear \label{iv}.
\end{enumerate}
\end{Prop}
As seen in Sec.~\ref{Hscrm}, it is usual to consider only hybrid probability 
measures $w$ such that $w(A,I)=0$ whenever $A$ is null for a reference 
measure. Within such an approach, the natural reference measure is the 
counting measure for discrete classical systems, i.e., with countable sample 
spaces \cite{PRA}, and the product measure for bipartite classical-hybrid and 
hybrid-hybrid systems. The above proposition also holds in this case since, in 
its proof, the probability measures built for bipartite systems vanish for the 
corresponding null events as soon as those of the hybrid system do. We 
discuss in the following transformations $\cal T$ with non-convex domains 
$\cal W$. In these cases, condition \eqref{nsc} or condition \eqref{nsc2} 
implies that $\cal T$ is convex-linear on any convex subset of $\cal W$.
\begin{proof} 
Let $t \in [0,1]$, $w_1$ and $w_2$ be any probability measures of the hybrid 
system $h$, $p_1$, $p_2$, $\eta_1$ and $\eta_2$ be the corresponding 
classical probability measures and maps from $X$ to the set of density 
operators on $\cal H$ given by Eq.\eqref{Gt}. Consider a quantum two-level 
system $q$ and denote $F_1$ and $F_2$ the one-dimensional orthogonal 
projectors summing to $I_q$ and $\cal G$ the set of effects on 
${\cal H}\otimes {\cal H}_q$. Define the maps 
$\tau_n : A \mapsto \int_A \eta_n dp_n \otimes F_n$ on $\cal A$ and 
$v_n : (A,G) \mapsto \mathrm{tr} ( \tau_n(A) G)$ on 
${\cal A} \times {\cal G}$. Since $\tau_n(A)$ is a positive trace-class operator 
for any event $A$, the linear form $T \mapsto \mathrm{tr}(T\otimes F_n G)$ 
on the set of trace-class operators on $\cal H$ is continuous for any bounded 
operator $G$, the map $G \mapsto \mathrm{tr}(\tau_n(A) G)$ on $\cal G$ 
is weakly sequentially continuous for any event $A$ \cite{PRA,BS}, and 
$\mathrm{tr} (\tau_n(X))=1$, $v_n$ is a probability measure of $hq$. It 
follows that so is $w=tv_1+(1-t)v_2$. Direct calculations lead to 
$w_{I_q}=tw_1+(1-t)w_2$, $w_q(F_n)=t+(1-2t)(n-1)$ and 
$w_{F_n}=w_n$ and so condition \eqref{nsc} gives 
${\cal T}(tw_1+(1-t)w_2)= t {\cal T}(w_1)+(1-t) {\cal T}(w_2)$.

Consider a classical system $c$ with the sample space $Y=\{ 1, 2 \}$ and the 
set of all subsets of $Y$ as the event space $\cal B$. Let 
${\cal C}=\{(A_1 \times \{ 1 \})\cup (A_2 \times \{ 2 \}) \vert
A_1,A_2 \in {\cal A} \}$. It can be shown, by direct calculations, that $\cal C$ 
is a $\sigma$-algebra and that all sets $A\times B$, with $A \in {\cal A}$ and 
$B \in {\cal B}$, belong to it. As ${\cal A} \otimes {\cal B}$ is closed under 
unions, $\cal C$ is a subset of it, and so ${\cal A} \otimes {\cal B}=\cal C$. 
Define the maps $v_y$ on ${\cal C}\times{\cal E}$, where $y \in Y$, as 
follows. For a given effect $E$, $C \mapsto v_y(C,E)$ is the product measure 
of $A \mapsto w_y(A,E)$ on $\cal A$ and of the Dirac measure $\delta_y$ on 
$\cal B$, i.e., such that $\delta_y(B)=1$ when $y \in B$ and zero otherwise, 
and so $v_y : (C,E) \mapsto w_y(C_y,E)$ with 
$C_y=\{x \in X | (x,y) \in C \}$ \cite{Bi}. The map $v_y$ is a probability 
measure of $hc$ as it fulfills all the required properties. Let 
$w=t v_1+(1-t)v_2$. For any $C \in {\cal A} \otimes {\cal B}$, there are 
$A_1$ and $A_2$ in ${\cal A}$ such that 
$C=(A_1 \times \{ 1 \})\cup (A_2 \times \{ 2 \})$, and so one has 
$w(C,E)=tw_1(A_1,E)+(1-t)w_2(A_2,E)$. It follows that
$w_{Y}=tw_1+(1-t)w_2$, $p_c(\{ y \})=t+(1-2t)(y-1)$ and 
$w_{\{ y \}}=w_y$, and thus condition \eqref{nsc2} gives 
${\cal T}(tw_1+(1-t)w_2)= t {\cal T}(w_1)+(1-t) {\cal T}(w_2)$.

We have shown that \ref{i} implies \ref{iv} and that \ref{ii} implies \ref{iv}. 
As seen in Sec.~\ref{Ahs}, \ref{iii} gives \ref{i} and \ref{ii}. As discussed 
above, \ref{iv} leads to the three other assertions. Consequently, all assertions 
are equivalent.
\end{proof}
\section{No quantum reaction on classical trajectories}\label{Nqb}
We start our study of the consequences of the no-signaling conditions 
\eqref{nsc}--\eqref{nsc3} for specific hybrid dynamics by considering those 
with usual classical trajectories. More precisely, in this section, we assume 
Dirac measures $p=\delta_{x}$ for the classical degrees of freedom, i.e., such 
that $\delta_{x}(A)=1$ when $x$ belongs to the event $A$ and zero otherwise. 
Within such hybrid approaches, there is no classical diffusion, unlike, for 
example, in Ref.\cite{OSSW}. Let us first show the following proposition.
\begin{Prop}\label{ct}
Let $x$ be any element of the sample space $X$ of a hybrid system $h$, 
${\cal W}_x$ be the set of probability measures of $h$ given by Eq.\eqref{Gt} 
with $p=\delta_{x}$, and $\cal T$ be any probability measure transformation 
of $h$ that is convex-linear on ${\cal W}_x$.

If the transformed classical probability measure $A \mapsto {\cal T}(w)(A,I)$ 
is a Dirac measure for any $w$ in ${\cal W}_x$, then it is identical for all $w$ 
in ${\cal W}_x$.
\end{Prop}
\begin{proof}
Let $\eta_1$ and $\eta_2$ be any $p$-integrable maps from $X$ to the set of 
density operators on $\cal H$, $\eta_3=(\eta_1+\eta_2)/2$, 
$w_{s} : (A,E) \mapsto \int_A \mathrm{tr}(\eta_s(x') E) dp(x')$ and 
$p'_s : A \mapsto {\cal T}(w_s)(A,I)$. By assumption, as the probability 
measures $w_s$ belong to ${\cal W}_x$, there are $x_s$ such that 
$p'_s=\delta_{x_s}$. On the other hand, the convex-linearity of $\cal T$ on 
${\cal W}_x$ implies $p'_3=(p'_1+p'_2)/2$. Thus, one has, for any event $A$, 
$2\delta_{x_3}(A)=p'_1(A)+p'_2(A)$ and so, as $p'_1$ and $p'_2$ are maps 
to the unit interval, $p'_1(A)=p'_2(A)=0$ when $x_3 \notin A$ and 
$p'_1(A)=p'_2(A)=1$ when $x_3 \in A$, i.e., $p'_s=\delta_{x_3}$. Hence, 
$p'_1$ and $p'_2$ are equal to each other, and thus the probability measure 
$A \mapsto {\cal T}(w)(A,I)$ is the same for all $w \in {\cal W}_x$.
\end{proof}
For $p=\delta_{x}$, Eq.\eqref{Gt} simplifies into 
$w : (A,E) \mapsto \mathrm{tr}(\rho E) \delta_{x}(A)$, where $\rho$ is 
the state of $q_h$, which is given by $\rho=\eta(x)$. Note that since for 
a given $w$, $\eta$ is $p$-almost everywhere unique, the maps $\eta$ 
corresponding to $w$ are all those such that $\eta(x)=\rho$. 
For the particular hybrid approaches considered in this section, the classical 
state $x$, e.g.,  position and  momentum, and the quantum state $\rho$ can 
be regarded, together, as the state of the hybrid system $h$. Within a 
description of hybrid systems taking into account all quantum states but only 
Dirac measures for classical systems, the set of probability measures of $h$ is 
${\cal W}=\cup_{x \in X} {\cal W}_x$. Any set ${\cal W}_x$ is convex. 
Moreover, any convex subset of ${\cal W}$ is a subset of a set ${\cal W}_x$ 
since $(w_1+w_2)/2$, for instance, does not belong to ${\cal W}$ when 
$w_1 \in {\cal W}_x$ and $w_2 \in {\cal W}_{x'}$ with $x \neq x'$. A 
transformation $\cal T$ obeying the no-signaling condition \eqref{nsc} is 
convex-linear on any subset ${\cal W}_x$, see the proof of Proposition 
\ref{Clr}, and so on any convex subset of ${\cal W}$.

Consequently, for hybrid dynamics with classical trajectories and all quantum 
states fulfilling the no-signaling condition \eqref{nsc}, Proposition \ref{ct} 
implies that any finite-time evolution is described by a hybrid state 
transformation of the form $(x,\rho) \mapsto (\phi(x),T_x(\rho))$, with $\phi$ 
a map from $X$ to itself and $T_x$ an $x$-dependent convex-linear map from 
the set of density operators on $\cal H$ to itself. The time evolution of the 
classical system $c_h$ is not influenced by the quantum system $q_h$, 
whereas that of $q_h$ can be influenced by $c_h$. In other words, there can 
be a classical action but no quantum reaction for such dynamics. 
We reiterate that this absence of quantum reaction is a consequence of 
both the no-signaling condition \eqref{nsc} and the restriction of the classical 
degrees of freedom to Dirac measures. Condition \eqref{nsc} can be rewritten 
here as $T_x(\rho_{I_q})= \sum_n w_q(F_n) T_x(\rho_{F_n})$ for any $x$, 
with $w_q(F)=\mathrm{tr} (\Gamma I \otimes F)$ and 
$\rho_{F}=\mathrm{tr}_q (\Gamma I \otimes F)/w_q(F)$, where $F$ is any 
effect of $q$ and $\Gamma$ is the state of the quantum system made up of 
$q_h$ and $q$, since all conditional probability measures $w_F$ given by 
Eq.\eqref{wF} belong to the same set ${\cal W}_x$ and 
${\cal T}(w_F) : (A,E) \mapsto 
\mathrm{tr}(T_x(\rho_{F})E)\delta_{\phi(x)}(A)$. As $T_x$ obviously fulfills 
the above condition when it is convex-linear, the above form for state 
transformations and  the no-signaling condition \eqref{nsc} are equivalent for 
hybrid dynamics with classical trajectories and all quantum states.

Condition \eqref{nsc2} is trivially fulfilled when only Dirac measures are 
considered for classical systems. This can be seen as follows. The probability 
measure of the bipartite system $hc$ reads as
$w : (C,E) \mapsto \mathrm{tr}(\rho E) \delta_{(x,y)}(C)$, 
with $x \in X$ and $y \in Y$. The corresponding probability measure of $c$ is 
$\delta_{y}$ and $w_B=w_Y$ for any event $B$ of $c$ such that 
$\delta_{y}(B)>0$. So, for any measurement performed on $c$, only one 
outcome $\tilde n$ is possible, the one corresponding to the event 
$B_{\tilde n}$ containing $y$, and $w_{B_{\tilde n}}=w_Y$. Consequently, 
the sum in Eq.\eqref{nsc2} reduces to the term with $n=\tilde n$, which is 
equal to the left-hand side of the equation. Similar arguments show that 
condition \eqref{nsc3} becomes condition \eqref{nsc} for $hq$.

Many classical-quantum models are commonly used in which classical degrees 
of freedom $x$, e.g., electromagnetic ones, evolve with time $t$ independently 
of quantum degrees of freedom whose evolution is governed by 
a time-dependent Hamiltonian $H(x(t))$ \cite{CDG}. Within such an approach, 
any finite-time evolution of the hybrid system is described by a transformation 
$(x,\rho) \mapsto (\phi(x),U_x\rho U_x^\dag)$, with $\phi$ a map from $X$ 
to itself and $U_x$ an $x$-dependent unitary operator on $\cal H$, which is of 
the form given above. As a second example, consider the Born-Oppenheimer 
molecular dynamics approach in which the nuclei are treated classically and 
the electrons quantum-mechanically. Here, contrary to the above-discussed 
hybrid dynamics, not all density operators are allowed for the electronic system, 
which has no true dynamics. It is assumed to always be in the ground state 
$| \psi (r) \rangle$ of a Hamiltonian $H(r)$ that depends on the nuclei positions 
$r$. These positions evolve according to the usual Newtonian equations with 
forces derived from the $r$-dependent ground energy of $H(r)$ \cite{MH}. 
The set of probability measures of the hybrid nuclei-electrons system reads as 
${\cal W}=\cup_{x \in X} {\cal W}_x$, where ${\cal W}_x$ consists of the 
single element $w : (A,E) \mapsto \langle \psi(r) | E | \psi (r) \rangle 
\delta_x(A)$ and $x$ denote the positions and momenta of the nuclei. 
Any finite-time evolution is described by a transformation 
$(x,\eta(x)) \mapsto (\phi(x),\eta \circ \phi(x))$ with $\phi$ a map from $X$ to 
itself and $\eta : x \mapsto | \psi (r) \rangle \langle \psi(r) |$. In this case, 
condition \eqref{nsc3} is trivially satisfied since the only possible states for the 
quantum system made up of $q_h$ and $q$ are product states, and so the 
conditional probability measures $w_{B_m,F_n}$ appearing in Eq.\eqref{nsc3} 
are all equal to $w_{Y,I_q}$. 

Within some descriptions of hybrid systems, Dirac measures are not possible 
classical probability measures. This is the case, for instance, if $X$ is a 
Euclidean space, $\mu$ is the usual Lebesgue measure, and only hybrid 
probability measures $w$ such that $w(A,I)=0$ for any $\mu$-null event $A$ 
are allowed, see Sec.~\ref{Hscrm}. However, within such approaches, classical 
localized states can be described approximately if, for any element $x$ of $X$, 
there is a sequence of classical probability measures $p_{n}$ such that 
$\lim_{n \rightarrow \infty} p_n(A)=\delta_x (A)$ for any event A. The 
sequence of probabilities $w_n(A,E)$ given by Eq.\eqref{Gt} with $p=p_n$ 
then converges to $\mathrm{tr}(\eta(x) E) \delta_{x}(A)$ for any event A 
and effect $E$ \cite{fn}. We say that $(p_n)_n$ is a Dirac sequence at $x$. 
We prove below an analog of Proposition \ref{ct} for Dirac sequences. For 
that purpose, let us first observe that Eq.\eqref{Gt} gives correct hybrid 
probability measures for maps $\eta$ that do not depend on $p$, as shown 
by the following lemma.
\begin{Lem}
Let $\eta$ be any map from $X$ to the set of density operators on $\cal H$ 
and $p$ be any probability measure on $\cal A$. 

The map $\eta$ is such that $x \mapsto \mathrm{tr}(\eta(x) M)$ is 
$p$-integrable for any bounded operator $M$ if and only if it is such that 
$x \mapsto \mathrm{tr} (\eta(x)E)$ is measurable for any effect $E$.
\end{Lem}
\begin{proof}
Assume that $\eta$ is such that $x \mapsto \mathrm{tr} (\eta(x)E)$ is 
measurable for any effect $E$. For any positive operator $P$, there is an effect 
$E$ such that $P=\Vert \sqrt{P} \Vert_{op}^2 E$, where 
$\Vert \cdot \Vert_{op}$ denotes the operator norm, and so
$x \mapsto \mathrm{tr} (\eta(x)P)$ is measurable. Let $M$ be any bounded 
operator. It can be written as $M=P_1-P_2+iP_3-iP_4$ with positive operators 
$P_r$ \cite{BS}. The map $g_{r} : x \mapsto \mathrm{tr}(\eta(x) P_r)$ is 
non-negative and, as seen above, measurable and so $I_{r}=\int g_{r} dp$ is 
well-defined. Moreover, as $g_{r}$ is upper-bounded by the operator norm of 
$P_r$ and $p(X)=1$, $I_{r}$ is finite. Thus, $g_{r}$ is $p$-integrable, and so 
is $g_{1}-g_{2}+ig_{3}-ig_{4} : x \mapsto \mathrm{tr} (\eta(x)M)$. The 
converse is trivial.
\end{proof}
We call any map $\eta$ from $X$ to the set of density operators on $\cal H$ 
such that $x \mapsto \mathrm{tr} (\eta(x)E)$ is measurable for any effect $E$ 
a quantum state map. 
\begin{Prop}
Let $x$ be any element of the sample space $X$ of a hybrid system $h$, 
${\cal D}_x$ be any convex set of Dirac sequences at $x$, ${\cal W}_x$ be the 
set of probability measures of $h$ given by Eq.\eqref{Gt} with $p$ equal to any
element of any sequence of ${\cal D}_x$, and $\cal T$ be any probability 
measure transformation of $h$ that is convex-linear on any convex subset of 
${\cal W}_x$.

If, for any sequence of hybrid probability measures 
$w_n~:~(A,E) \mapsto \int_A \mathrm{tr}(\eta(x)E) dp_n(x)$ with 
$(p_n)_n \in {\cal D}_x$ and $\eta$ a quantum state map, there is $x'$ 
such that $(p'_n)_n$ is a Dirac sequence at $x'$, where 
$p'_n : A \mapsto {\cal T}(w_n)(A,I)$, then $x'$ is the same for any such 
sequence $(w_n)_n$.
\end{Prop}
\begin{proof}
Let $\eta_1$ and $\eta_2$ be any quantum state maps, 
$\eta_3=(\eta_1+\eta_2)/2$, and $(p_n)_n$ be any sequence of ${\cal D}_x$. 
We define $w_{s,n} : (A,E) \mapsto \int_A \mathrm{tr}(\eta_s(x) E) dp_n(x)$ 
and $p'_{s,n} : A \mapsto {\cal T}(w_{s,n})(A,I)$. The probability measures 
$w_{s,n}$ and $p'_{s,n}$ depend {\it a priori} on $p_n$ and $\eta_s$. The 
convex-linearity of $\cal T$ on any convex subset of ${\cal W}_x$ implies
$p'_{3,n}=(p'_{1,n}+p'_{2,n})/2$. By assumption, there are elements $x_s$ 
of $X$ such that, for any event $A$, 
$\lim_{n \rightarrow \infty} p'_{s,n}(A)=\delta_{x_s}(A)$. Thus, one has 
$2\delta_{x_3}(A)=\delta_{x_1}(A)+\delta_{x_2}(A)$, which gives $x_s=x_3$. 
Consequently, $x_1$ and $x_2$ are equal to each other and so depend only 
on $(p_n)_n$.

Let $(p_{1,n})_n$ and $(p_{2,n})_n$ be any elements of ${\cal D}_x$, 
$(p_{3,n})_n=((p_{1,n}+p_{2,n})/2)_n$, and $\eta$ be any quantum state 
map. We define 
$w_{s,n} : (A,E) \mapsto \int_A \mathrm{tr}(\eta(x) E) dp_{s,n}(x)$ and 
$p'_{s,n} : A \mapsto {\cal T}(w_{s,n})(A,I)$. Since, for any measurable 
non-negative map $g$ on $X$, 
$\int g d p_{3,n}=\int g d p_{1,n}/2+\int g d p_{2,n}/2$, one has 
$w_{3,n}=(w_{1,n}+w_{2,n})/2$. By assumption, as the sequences 
$(p_{s,n})_n$ belong to ${\cal D}_x$, there are elements $x_s$ of $X$ such 
that, for any event $A$, 
$\lim_{n \rightarrow \infty} p'_{s,n}(A)=\delta_{x_s}(A)$. We have shown 
above that $x_s$ depends only on $(p_{s,n})_n$. The convex-linearity of 
$\cal T$ on any convex subset of ${\cal W}_x$ implies 
$p'_{3,n}=(p'_{1,n}+p'_{2,n})/2$. Arguing as above gives $x_1=x_2$, 
which finishes the proof.
\end{proof}
\section{No classical reaction under quantum purity preservation}
\label{Ncb}
In this section, we focus on specific no-signaling hybrid dynamics for which 
pure states of quantum systems remain pure. More precisely, we assume 
that the quantum degrees of freedom are described by a pure state at all times, 
provided their initial state is pure. We say that hybrid probability measures of 
the form $(A,E) \mapsto \mathrm{tr}(\rho E) p(A)$, where $\rho$ is a quantum 
state and $p$ is a classical probability measure, are uncorrelated, as there are no 
statistical correlations between the classical and the quantum degrees of freedom 
for such probability measures. The following two lemmas are useful for the proofs 
of Propositions \ref{ps} and \ref{Ue}.
\begin{Lem}\label{Lemps}
The probability measure of a hybrid system is uncorrelated whenever the state 
of its quantum subsystem is pure. 
\end{Lem}
This lemma shows, for hybrid systems, the analog of the well-known fact that 
there is no correlation between two quantum systems when one of them is in 
a pure state.
\begin{proof}
For any pure state $\rho$ of the quantum subsystem of a hybrid system $h$, 
one finds 
$\int (1-\mathrm{tr}(\rho \eta(x)) ) dp(x)=1-\mathrm{tr} (\rho^2)=0$, using 
Eq.\eqref{Gt} and the continuity of the linear form 
$T \mapsto \mathrm{tr}(\rho T)$ on the set of trace-class operators $T$ on 
$\cal H$ \cite{BI}. As $1-\mathrm{tr}(\rho \eta(x))$ is non-negative for any 
$x$, the above equality implies that there is a $p$-null set $A$ such that, for 
any $x \notin A$, $\mathrm{tr}(\rho \eta(x))=1$, and so 
$\mathrm{tr}((\rho - \eta(x))^2)=\mathrm{tr}( \eta(x)^2)-1$. Since this 
quantity is non-negative, one has $\mathrm{tr}( \eta(x)^2)=1$, and hence 
$\eta(x)=\rho$. Thus, for any effect $E$, 
$\mathrm{tr} (\eta(x)E)=\mathrm{tr} (\rho E)$ for $p$-almost every $x$, and 
so the probability measure of $h$ is uncorrelated, see Eq.\eqref{Gt}.
\end{proof}
When a hybrid probability measure transformation $\cal T$ is convex-linear on 
a convex set $\cal W$, it fulfills 
${\cal T}(\sum_n t_n w_n)=\sum_n t_n {\cal T}(w_n)$ for any finite 
sequences of non-negative numbers $t_n$ summing to unity and probability 
measures $w_n$ of $\cal W$. The following lemma extends this result to 
infinite sequences.
\begin{Lem}\label{Lemiclc}
Let $\cal W$ be any convex set of hybrid probability measures sequentially 
closed with respect to pointwise convergence, $(t_n)_{n}$ and $(w_n)_{n}$ be 
any infinite sequences of, respectively, non-negative numbers summing to unity 
and elements of $\cal W$, and $\cal T$ be any probability measure 
transformation that is convex-linear on $\cal W$.

The sum $w(A,E)=\sum_n t_n w_n(A,E)$ converges for any event $A$ and 
effect $E$, the map $w$ so defined belongs to $\cal W$, and ${\cal T}(w)$ is 
given by ${\cal T}(w)(A,E)=\sum_n t_n {\cal T}(w_n)(A,E)$.
\end{Lem}
\begin{proof} 
Let $W_N=\sum_{n \le N} t_n w_n(A,E)$ and $T_N=\sum_{n \le N} t_n$ for 
any $N$, event $A$ and effect $E$. As $|W_{N'}-W_N| \le |T_{N'}-T_N|$ for 
any $N$ and $N'$, $(W_N)_N$ is a Cauchy sequence and so converges. Since 
$W_N \ge 0$ for any $N$, the values of $w$ are non-negative. Consider any 
effect $E$ and any sequence of pairwise-disjoint events $A_m$, and let 
$A=\cup_m A_m$. The properties of the probability measures $w_n$ lead to 
$w(A,E)=\sum_n \sum_m t_n w_n (A_m,E)$. As the summands are 
non-negative, the sums can be interchanged and so 
$w(A,E)=\sum_m w(A_m,E)$. Similarly, one has $w(A,E)=\sum_n w(A,E_n)$ 
for any event $A$ and any sequence of effects $E_n$ such that 
$E=\sum_n E_n$ is an effect. Since, moreover, $w(X,I)=\sum_n t_n=1$, $w$ 
is a probability measure. For any $N$, define the probability measure 
$w'_{N}=\sum_{n \le N} t_n w_n/T_N$. It belongs to $\cal W$ as this set is 
convex. Moreover, the sequence $(w'_{N})_N$ pointwise converges to $w$, 
which thus belongs to $\cal W$ as it is sequentially closed with respect to 
pointwise convergence.

For any $N$, let $w''_{N}=(w-T_N w'_N)/(1-T_N)$. By construction, the values 
of $w''_{N}$ are non-negative and $w''_{N}(X,I)=1$. Since, furthermore, the 
maps $w$ and $w'_N$ are probability measures, so is $w''_{N}$. For any $N$ 
and $N'>N$, we define the probability measure 
$w''_{N,N'}=\sum_{n=N+1}^{N'} t_n w_n/(T_{N'}-T_N)$. It belongs to the 
convex set $\cal W$ and the sequence $(w''_{N,N'})_{N'}$ pointwise converges 
to $w''_N$, which thus belongs to the closed set $\cal W$. The convex-linearity 
of ${\cal T}$ on $\cal W$ gives ${\cal T}(w)
=\sum_{n \le N} t_n {\cal T}(w_n)+(1-T_N){\cal T}(w''_{N})$. Thus, for any 
event $A$ and effect $E$, 
$|{\cal T}(w)(A,E)-\sum_{n \le N} t_n {\cal T}(w_n)(A,E)|$ is upper-bounded 
by $1-T_N$ and so vanishes as $N \rightarrow \infty$.
\end{proof}
The following proposition can now be shown. We denote by $\hat \rho(w)$ the 
quantum subsystem state corresponding to the hybrid probability measure $w$, 
i.e., such that $\mathrm{tr} (\hat \rho(w) E)=w(X,E)$ for any effect $E$, and 
by ${\cal W}_{\rho}$ the set of probability measures $w$ such that 
$\hat \rho(w)=\rho$. 
\begin{Prop}\label{ps}
Let ${\cal W}_{uc}$ be the set of uncorrelated probability measures of a hybrid 
system $h$ and $\cal T$ be any probability measure transformation of $h$ that 
is convex-linear on any convex subset of ${\cal W}_{uc}$.

If $\cal T$ is such that $\hat \rho \circ {\cal T}(w)$ is pure for any $w$ for 
which $\hat \rho(w)$ is pure then, for any state $\rho$ of the quantum 
subsystem $q_h$ of $h$, $\hat \rho \circ {\cal T}(w)$ is the same for any 
$w$ of ${\cal W}_{\rho} \cap {\cal W}_{uc}$. 
\end{Prop}

As shown by Lemma \ref{Lemps}, when $\rho$ is pure, all probability measures 
of ${\cal W}_{\rho}$ are uncorrelated, i.e., 
${\cal W}_{\rho} \cap {\cal W}_{uc}={\cal W}_{\rho}$. We remark that this 
proposition also holds when only hybrid probability measures $w$ such that 
$A \mapsto w(A,I)$ belongs to a given convex set of classical probability 
measures are allowed, e.g., when $w(A,I)=0$ as soon as $A$ is null for a 
reference measure, see Sec.~\ref{Hscrm}.
\begin{proof}
Consider first the case of a pure state $\rho$ of $q_h$. Let $w_1$ and $w_2$ 
be any elements of ${\cal W}_{\rho}$ and $w_3=(w_1+w_2)/2$. This 
probability measure belongs to ${\cal W}_{\rho}$ since 
$\mathrm{tr} (\hat \rho(w_3) E)=\mathrm{tr} (\rho E)$ for any effect $E$. 
Similarly, $tw_1+(1-t)w_2$ belongs to ${\cal W}_{\rho}$ for any $t \in [0,1]$ 
and so ${\cal W}_{\rho}$ is convex. Denote 
$\rho'_s=\hat \rho \circ {\cal T}(w_s)$. By assumption, the states 
$\rho'_s=|s \rangle \langle s |$ are pure, where $|s \rangle$ are normalized 
vectors of $\cal H$. As ${\cal T}(w_{s})(X,E)= \mathrm{tr}(\rho'_s E)$ for 
any effect $E$ and $\cal T$ is convex-linear on ${\cal W}_{\rho}$, which is 
a subset of ${\cal W}_{uc}$, as shown by Lemma \ref{Lemps}, one has 
$\rho'_3=(\rho'_1+\rho'_2)/2$. This implies 
$|\langle 3 | 1 \rangle|^2 + |\langle 3 | 2 \rangle|^2 =2$ which leads to 
$| s \rangle=z_s | 3 \rangle$ with $|z_s|=1$. Thus, $\rho'_1$ and $\rho'_2$ 
are equal to each other and so $\hat \rho \circ {\cal T}(w)$ is the same for 
any $w$ of ${\cal W}_{\rho}$. 

Let $p$ be any probability measure of the classical subsystem $c_h$ of $h$ and 
${\cal W}_p^{uc}$ be the set of uncorrelated probability measures 
$(A,E) \mapsto \mathrm{tr}(\rho E) p(A)$ with $\rho$ any state of $q_h$. It is 
a convex subset of ${\cal W}_{uc}$. Let us show that it is sequentially closed 
with respect to pointwise convergence. Consider any sequence $(w_n)_n$ of 
${\cal W}_p^{uc}$ pointwise converging to a probability measure $w$. Let 
$\rho = \hat \rho (w)$. For any event $A$ and effect $E$, 
$|\mathrm{tr} (\rho E) p(A)-w_n(A,E)|
=p(A) |w(X,E)-w_n(X,E)|$ vanishes as $n$ goes to infinity and so 
$w(A,E)=\mathrm{tr} (\rho E) p(A)$, i.e., $w \in {\cal W}_p^{uc}$.

Define the map $\zeta : \rho \mapsto \hat \rho \circ {\cal T}(w)$, where 
$w : (A,E) \mapsto \mathrm{tr}(\rho E)p(A)$, from the set of density 
operators on $\cal H$ to itself. It may {\it a priori} depend on $p$. For any 
given state $\rho$, the values of $w$ can be written as 
$w(A,E)=\sum_n \lambda_n w_n(A,E)$ where $\lambda_n$ denotes the 
eigenvalues of $\rho$, which are non-negative and sum to unity, and the maps 
$w_n : (A,E) \mapsto \langle n | E | n \rangle p(A)$, with $| n \rangle$ 
denoting the eigenstates of $\rho$, belong to ${\cal W}_p^{uc}$. It follows 
from Lemma \ref{Lemiclc} that, for any effect $E$, 
${\cal T}(w)(X,E)=\sum_n \lambda_n {\cal T}(w_n)(X,E)$ and so, by definition 
of $\zeta$, $\mathrm{tr}(\zeta(\rho)E)=
\sum_n \lambda_n\mathrm{tr}(\rho'_n E)$ with 
$\rho'_n=\hat \rho \circ {\cal T}(w_n)$. As shown above, the states $\rho'_n$ 
are independent of $p$ and so is $\zeta$. For any state $\rho$ of $q_h$, 
$\hat \rho \circ {\cal T}(w)=\zeta(\rho)$ is the same for any 
$w$ of ${\cal W}_{\rho} \cap {\cal W}_{uc}$.
\end{proof}
When the density operator $\rho=|\psi \rangle \langle \psi |$ of the quantum 
system $q_h$ is pure, it follows from Lemma \ref{Lemps} that Eq.\eqref{Gt} 
simplifies into $w : (A,E)\mapsto \langle \psi | E | \psi \rangle p(A)$. 
Thus, for the particular hybrid approaches in which the quantum degrees 
of freedom are described by a pure state, the classical probability measure 
$p$ and the unit ray 
$\psi=\{ e^{i \theta} | \psi \rangle | \theta \in  \mathbb{R} \}$ can be 
considered, together, to be the state of the hybrid system $h$. When only pure 
quantum states are taken into account, the set of probability measures of $h$ 
is ${\cal W}=\cup_{\psi} {\cal W}_{\psi}$ where $\psi$ runs over all unit rays 
of $q_h$ and ${\cal W}_{\psi}={\cal W}_{|\psi \rangle \langle \psi |}$. Any set 
${\cal W}_{\psi}$ is convex. Moreover, any convex subset of ${\cal W}$ is a 
subset of a set ${\cal W}_{\psi}$ since $(w_1+w_2)/2$, for instance, does not 
belong to ${\cal W}$ when $w_1 \in {\cal W}_{\psi}$ and 
$w_2 \in {\cal W}_{\psi'}$ with $\psi \neq \psi'$. A transformation $\cal T$ 
obeying the no-signaling condition \eqref{nsc2} is convex-linear on any subset 
${\cal W}_{\psi}$, see the proof of Proposition \ref{Clr}, and so on any convex 
subset of ${\cal W}$. 

The convex-linearity of $\cal T$ on ${\cal W}_{\psi}$ is sufficient to show that 
$\hat \rho \circ {\cal T}(w)$ is the same for any $w \in {\cal W}_{\psi}$, see 
the proof of Proposition \ref{ps}. Consequently, for hybrid dynamics with only 
pure quantum states and a convex set of classical probability measures satisfying 
the no-signaling condition \eqref{nsc2}, any finite-time evolution is described by 
a hybrid state transformation of the form 
$(p,\psi) \mapsto (T_{\psi}(p),\Phi(\psi))$, with $\Phi$ a map from the set of 
unit rays to itself and $T_{\psi}$ a $\psi$-dependent convex-linear map from 
the set of probability measures of $c_h$ to itself. The time evolution of the 
quantum system $q_h$ is not influenced by the classical system $c_h$, whereas 
that of $c_h$ can be influenced by $q_h$. In other words, there can be a 
quantum action but no classical reaction for such dynamics. We reiterate 
that this absence of classical reaction is a consequence of both the 
no-signaling condition \eqref{nsc2} and the restriction of the quantum degrees 
of freedom to pure states. Condition \eqref{nsc2} can be rewritten here as 
$T_{\psi}(p_Y)= \sum_n p_c(B_n) T_{\psi}(p_{B_n})$ for any $\psi$, with 
$p_c(B)=p(X \times B)$ and $p_B : A \mapsto p(A \times B)/p_c(B)$, where 
$B$ is any event of $c$ and $p$ is the probability measure of the classical 
system made up of $c_h$ and $c$. So, for hybrid dynamics with only pure 
quantum states and a convex set of classical probability measures, the 
above form for state transformations and the no-signaling condition 
\eqref{nsc2} are equivalent.

In contrast to the Dirac measures discussed in the previous section, the state of 
a composite quantum system can be pure while those of its subsystems are not. 
If strictly only pure quantum states are allowed, only pure product states are 
possible for composite quantum systems. In this case, condition \eqref{nsc} is 
obviously fulfilled since the probability measure $w_F$ given by Eq.\eqref{wF} is 
equal to $w_{I_q}$ for any effect $F$ such that $w_q(F)>0$. When correlated 
quantum states are considered, and so mixed states for $h$, it follows from 
Proposition \ref{Clr} that the no-signaling conditions \eqref{nsc} and 
\eqref{nsc2} are equivalent and that a transformation $\cal T$ obeying them is 
convex-linear. So, for no-signaling classical-quantum dynamics such that pure 
states of quantum systems remain pure, the time evolution of $q_h$ is not 
influenced by $c_h$ whenever these two systems are initially uncorrelated, see 
Proposition \ref{ps}. 

Pure states of quantum systems remain pure if, for instance, they are assumed 
to evolve unitarily as, e.g., in pilot-wave approaches \cite{OM}. When only pure 
quantum states are taken into account, finite-time evolutions of the form 
discussed above are possible, and the quantum degrees of freedom can influence 
the classical ones. But if all quantum states are allowed, then the only 
no-signaling hybrid dynamics with unitary evolution for pure states are those for 
which the classical and quantum subsystems evolve independently of each 
other when they are initially uncorrelated, as shown by the proposition below.
\begin{Prop}\label{Ue}
Let ${\cal W}_{uc}$ be the set of uncorrelated probability measures of a hybrid 
system $h$ and $\cal T$ be any probability measure transformation of $h$ that 
is convex-linear on any convex subset of ${\cal W}_{uc}$.

If $\cal T$ is such that, for any $w$ for which $\hat \rho(w)$ is pure, 
$\hat \rho \circ {\cal T}(w)=U\hat \rho (w) U^\dag$, where $U$ is an isometry 
on $\cal H$, then, for any $w$ of ${\cal W}_{uc}$, ${\cal T}(w)$ is given by 
${\cal T}(w)(A,E)=\mathrm{tr}(U\hat \rho (w) U^\dag E)p(A)$, where the 
probability measure $p$ does not depend on $\hat \rho(w)$. 
\end{Prop}
\begin{proof}
Let $p$ be any probability measure of the classical subsystem $c_h$ of $h$ and 
${\cal W}_p^{uc}$ be the set of uncorrelated probability measures 
$(A,E) \mapsto \mathrm{tr}(\rho E) p(A)$ with $\rho$ any state of the quantum 
subsystem $q_h$ of $h$. It is a convex subset of ${\cal W}_{uc}$ and is 
sequentially closed with respect to pointwise convergence, see the proof of 
Proposition \ref{ps}. For any normalized vector $|\psi \rangle$ of $\cal H$, we 
denote by $p'_{\psi}$ the probability measure 
$A \mapsto {\cal T}(w_{\psi})(A,I)$, where $w_{\psi}$ is the element of 
${\cal W}_p^{uc}$ such that 
$\hat \rho(w_{\psi})=| \psi \rangle \langle \psi |$. It follows from the properties 
of $\cal T$, Lemma \ref{Lemps}, and Lemma \ref{Lemiclc} that, for any 
$w$ of ${\cal W}_p^{uc}$, ${\cal T}(w) : (A,E) \mapsto \sum_n \lambda_n 
\langle n | U^{\dag}E U | n \rangle p'_n(A)$, where $\lambda_n$ and 
$| n \rangle$ denote, respectively, the eigenvalues and eigenvectors of 
$\hat \rho(w)$. The values of ${\cal T}(w)$ can be rewritten as 
${\cal T}(w)(A,E)=\mathrm{tr} (\tau(A) U^{\dag}E U)$, where $\tau(A)$ is a 
positive trace-class operator on $\cal H$ whose eigenvectors do not depend 
on the event $A$.

Let $|1 \rangle$ and $|2 \rangle$ be any orthonormal vectors in $\cal H$. 
Denote by $w_1$, $w_2$ and $w$ the elements of ${\cal W}_p^{uc}$ such 
that $\hat \rho (w_n)=| n \rangle \langle n |$ and 
$\hat \rho (w)=(|1 \rangle \langle 1|+|2 \rangle \langle 2|)/4
+|\psi \rangle \langle \psi|/2$, where 
$|\psi \rangle=(|1 \rangle+|2 \rangle)/\sqrt{2}$. The properties of $\cal T$ 
and Lemma \ref{Lemps} give 
${\cal T}(w) : (A,E) \mapsto \mathrm{tr} (\tau(A) U^{\dag}E U)$
with $\tau(A) = ((p'_1(A)+p'_{\psi}(A))| 1 \rangle\langle 1 |
+(p'_2(A)+p'_{\psi}(A))| 2 \rangle\langle 2 |
+p'_{\psi}(A)| 1 \rangle\langle 2 |+p'_{\psi}(A)| 2 \rangle\langle 1 |)/4$. As 
seen above, the eigenvectors of $\tau(A)$ do not depend on $A$. With $A=X$, 
one finds the eigenvectors $|  \pm \rangle=| 1 \rangle \pm | 2 \rangle$, which 
are orthogonal to each other. It follows from 
$\langle + | \tau(A) | - \rangle=0$ for any event $A$ that $p'_1=p'_2$. 
Consequently, for any $w$ of ${\cal W}_{uc}$, ${\cal T}(w)$ is uncorrelated 
and $\hat \rho \circ {\cal T}(w)=U \hat \rho (w) U^\dag$, see the expression 
in the previous paragraph. So, due to Proposition \ref{uc}, the map 
$A \mapsto {\cal T}(w)(A,I)$ is the same for any $w$ of ${\cal W}_p^{uc}$.
\end{proof}
\section{Dynamics with no generation of classical-quantum correlations}
\label{Ngcqi}
In the previous two sections, restricted classical-quantum approaches were 
discussed, in which only specific uncorrelated probability measures are allowed. 
It was shown that, for no-signaling dynamics, there can be no quantum 
reaction assuming classical trajectories and no classical reaction assuming 
pure-state quantum trajectories. The proofs of these results rely on the 
particular assumptions considered in these cases. When the set of permitted 
hybrid probability measures contains all uncorrelated ones, the following 
proposition can be shown under the assumption that the dynamics generates 
no correlations between the classical and the quantum degrees of freedom. We 
denote by ${\cal W}_{uc}$ the set of uncorrelated probability measures of the 
hybrid system $h$, by ${\cal W}_p^{uc}$ the set of probability measures $w$ 
in ${\cal W}_{uc}$ such that $A \mapsto w(A,I)$ is equal to the probability 
measure $p$ of $c_h$ and by ${\cal W}_{\rho}^{uc}$ the set of probability 
measures $w$ in ${\cal W}_{uc}$ such that $\hat \rho(w)$ is equal to the state 
$\rho$ of $q_h$. Note that any convex subset of ${\cal W}_{uc}$ is a subset of 
a set ${\cal W}_p^{uc}$ or of a set ${\cal W}_{\rho}^{uc}$. 
\begin{Prop}\label{uc}
Let $\cal T$ be any probability measure transformation of a hybrid system $h$ 
that is convex-linear on any convex subset of the set ${\cal W}_{uc}$ of $h$. 
For any probability measure $w$ of $h$, $\hat p_{{}_{\cal T}}(w)$ denotes 
the classical probability measure $A \mapsto {\cal T}(w)(A,I)$, and 
$\hat \rho_{{}_{\cal T}}(w)=\hat \rho \circ {\cal T}(w)$.

If $\cal T$ is  such that ${\cal T}[{\cal W}_{uc}] \subseteq {\cal W}_{uc}$, 
then, on any set ${\cal W}_p^{uc}$ (${\cal W}_{\rho}^{uc}$), 
$\hat p_{{}_{\cal T}}$ is constant  or $\hat \rho_{{}_{\cal T}}$ is constant.
\end{Prop}
This proposition also holds when only hybrid probability measures $w$ such 
that $A \mapsto w(A,I)$ belongs to a given convex set of classical probability 
measures are considered, e.g., when $w(A,I)=0$ as soon as $A$ is null for a 
reference measure, see Sec.~\ref{Hscrm}. In contrast with Proposition 
\ref{uc}, when only Dirac measures are allowed for classical systems, both 
$\hat p_{{}_{\cal T}}$ and $\hat \rho_{{}_{\cal T}}$ can vary as the initial 
quantum subsystem state is fixed, see Sec.~\ref{Nqb}, and when only pure 
states are allowed for quantum systems, both $\hat p_{{}_{\cal T}}$ and 
$\hat \rho_{{}_{\cal T}}$ can vary as the initial classical subsystem probability 
measure is fixed, see Sec.~\ref{Ncb}.
\begin{proof}
Let $w_1$ and $w_2$ be any elements of ${\cal W}_p^{uc}$ and 
$w_3=(w_1+w_2)/2$ which belongs to ${\cal W}_p^{uc}$. The convex-linearity 
of ${\cal T}$ on ${\cal W}_p^{uc}$ gives 
${\cal T}(w_3)=({\cal T}(w_1)+{\cal T}(w_2))/2$ and so 
$p'_3=(p'_1+p'_2)/2$, where $p'_s=\hat p_{{}_{\cal T}}(w_s)$, and 
$\rho'_3=(\rho'_1+\rho'_2)/2$, where $\rho'_s=\hat \rho_{{}_{\cal T}}(w_s)$. 
By assumption, ${\cal T}(w_s)$ is given by 
${\cal T}(w_s)(A,E)=\mathrm{tr}(\rho'_sE) p'_s(A)$. Consequently, the above 
equalities lead to $\mathrm{tr}((\rho'_1-\rho'_2)E) (p'_1-p'_2)(A)=0$ for any 
event $A$ and effect $E$, and so $p'_1=p'_2$ or $\rho'_1=\rho'_2$.

Assume that $p$ is such that there exist elements $w_1$ and $w_2$ of 
${\cal W}_p^{uc}$ for which 
$\hat \rho_{{}_{\cal T}}(w_1) \neq \hat \rho_{{}_{\cal T}}(w_2)$. It follows 
from the above that $\hat p_{{}_{\cal T}}(w_1)=\hat p_{{}_{\cal T}}(w_2)$.
Let ${\cal V} = \{w \in {\cal W}_p^{uc} : 
\hat \rho_{{}_{\cal T}}(w)=\hat \rho_{{}_{\cal T}}(w_1) \}$. Any 
$w \in \cal V$ fulfills $\hat p_{{}_{\cal T}}(w)=\hat p_{{}_{\cal T}}(w_2)$ 
since $\hat \rho_{{}_{\cal T}}(w) \neq \hat \rho_{{}_{\cal T}}(w_2)$. Any 
$w \notin \cal V$ fulfills $\hat p_{{}_{\cal T}}(w)=\hat p_{{}_{\cal T}}(w_1)$ 
since $\hat \rho_{{}_{\cal T}}(w) \neq \hat \rho_{{}_{\cal T}}(w_1)$. 
Thus, in this case, $\hat p_{{}_{\cal T}}(w)$ is the same for any 
$w \in {\cal W}_p^{uc}$. In the alternative case, $\hat \rho_{{}_{\cal T}}$ is 
constant on ${\cal W}_p^{uc}$.

The same arguments hold with ${\cal W}_p^{uc}$ replaced by 
${\cal W}_{\rho}^{uc}$.
\end{proof}

To discuss the implications of Proposition \ref{uc}, let us 
first examine a particular case. Assume that the considered dynamics 
is no-signaling, does not generate classical-quantum correlations, and is, 
furthermore, time-translation symmetric. Such a dynamics is described by a 
time-parameterized family of convex-linear probability measure 
transformations ${\cal T}(t)$ such that ${\cal T}(t)(w)$ is uncorrelated 
whenever $w$ is, ${\cal T}(t'+t)={\cal T}(t') \circ {\cal T}(t)$ for any positive 
times $t$ and $t'$ and ${\cal T}(0) : w \mapsto w$. Let $w$ be any 
uncorrelated probability measure with probability measure $p_0$ for $c_h$ and 
state $\rho_0$ for $q_h$. Suppose additionally that for a short enough 
time $t$, the probability measure $p_t=\hat p_{{}_{{\cal T}(t)}}(w)$ of $c_h$ 
depends on $p_0$ and the state $\rho_t=\hat {\rho}_{{}_{{\cal T}(t)}}(w)$ of 
$q_h$ depends on $\rho_0$.  We stress that, although natural, this 
short-time behavior constitutes an additional assumption. In other words, for 
any $p_0$ and $\rho_0$, $\hat p_{{}_{{\cal T}(t)}}$ is not constant on 
${\cal W}_{{\rho}_0}^{uc}$ and $\hat \rho_{{}_{{\cal T}(t)}}$ is not constant 
on ${\cal W}_{p_0}^{uc}$. It follows from this assumption and Proposition 
\ref{uc} that, for any time $t$, $p_t$ does not depend on $\rho_0$ and $\rho_t$ 
does not depend on $p_0$. In other words, the two subsystems of $h$ evolve 
independently of one another. Consequently, for no-signaling time-translation 
symmetric dynamics, as soon as one subsystem influences the other, correlations 
are generated between them. 

In the Appendix, we present a measure of the correlations between a discrete 
classical system and a quantum system. It is upper-bounded by the von 
Neumann entropy of the state of $q_h$. So, when this quantum system is 
bipartite with a finite-dimensional Hilbert space, it can be shown that there is a 
tradeoff between the quantum entanglement of its two parts and its correlations 
with $c_h$ \cite{to1,to2,to3,to4}. Thus, classical-quantum interactions beyond 
those permitted by Proposition \ref{uc} have a detrimental impact on this 
entanglement. Convex-linear transformations $\cal T$ describing independent 
classical and quantum subsystems can be written explicitly for a hybrid system 
$h$ with a reference measure $\mu$. In this case, the probability measures $w$ of 
$h$ reads as 
$w : (A,E) \mapsto \int_A \mathrm{tr}(\omega(x) E) d\mu(x)$ where $\omega$ is 
a $\mu$-integrable map from $X$ to the set of positive trace-class operators on 
$\cal H$ such that $\mathrm{tr} \int  \omega d\mu=1$, see Sec.~\ref{Hscrm}. If 
$w$ is uncorrelated, one has $\omega : x \mapsto f(x)\rho$ with $\rho$ the state 
of $q_h$ and $f$ the probability density function of $c_h$. Consider a 
transformation $\cal T$ given by 
$${\cal T}(w) : (A,E) \mapsto \sum_{\alpha} \int k(A,x)
\mathrm{tr}\left(L_{\alpha}\omega(x) L_{\alpha}^\dag E \right)  d\mu(x), $$ 
where the map $k : {\cal A}\times X \rightarrow \mathbb{R}^+$ is a Markov 
kernel \cite{MK} such that $x \mapsto k(A,x)$ vanishes $\mu$-almost everywhere 
when $A$ is $\mu$-null and the bounded operators $L_{\alpha}$ are such that 
$\sum_{\alpha} L_{\alpha}^\dag L_{\alpha}=I$ \cite{PRA}. If $w$ is 
uncorrelated, so is ${\cal T}(w)$, with the probability measure 
$\hat p_{{}_{{\cal T}}}(w) : A \mapsto \int k(A,x) f(x) d\mu(x)$ of $c_h$, which 
depends only on $f$, and the state 
$\hat \rho_{{}_{{\cal T}}}(w)
=\sum_{\alpha} L_{\alpha} \rho L_{\alpha}^\dag$ of $q_h$, which depends only 
on $\rho$.

Proposition \ref{uc} is not restricted to convex-linear probability measure 
transformations describing independent classical and quantum subsystems, such 
as those of the time-translation symmetric dynamics discussed above. As an 
illustration, consider the two transformations below that do not generate 
classical-quantum correlations but transfer information from one subsystem to 
the other. For a hybrid system $h$ with a reference measure $\mu$, an example 
of a convex-linear transformation ${\cal T}$ such that, on any 
set ${\cal W}_p^{uc}$, $\hat \rho_{{}_{\cal T}}$  is constant but 
$\hat p_{{}_{\cal T}}$ is not, is given by 
$${\cal T}(w) : (A,E) \mapsto \mathrm{tr}(\rho' E)
\sum_{n=1}^N w(X,E_n)\frac{\mu(A \cap A_n)}{\mu(A_n)} , $$ where $\rho'$ 
is any density operator on $\cal H$, the $N$ effects $E_n$ sum to $I$ and the 
$N$ pairwise disjoint events $A_n$ are such that $\mu(A_n)$ is finite. Such a 
transformation corresponds to the following protocol. The measurement 
characterized by the effects $E_n$ is performed on $q_h$, the outcome of this 
measurement is stored using $c_h$ and the state of $q_h$ is always changed 
into $\rho'$. Similarly, a convex-linear transformation ${\cal T}$ such that, on 
any set ${\cal W}_{\rho}^{uc}$, $\hat p_{{}_{\cal T}}$ is constant but 
$\hat \rho_{{}_{\cal T}}$ is not, can be defined as $${\cal T}(w) : (A,E) \mapsto 
p'(A)\sum_{n=1}^N w(A_n,I) \langle n | E | n \rangle, $$ where $p'$ is any 
probability measure of $c_h$, $A_n$ denote pairwise disjoint events such that 
$\cup_{n=1}^N A_n =X$, and $| n \rangle$ denote orthonormal vectors in 
$\cal H$. Here, the measurement characterized by the events $A_n$ is performed 
on $c_h$, the outcome of this measurement is stored using $q_h$ and the 
probability measure of $c_h$ is changed into $p'$. 
\section{Conclusion}\label{C}
We have formulated a dynamical no-signaling condition for hybrid approaches 
and shown that it has important consequences for classical-quantum interactions. 
This requirement means that, when two systems evolve independently of each 
other, the outcome probabilities of a measurement made on one system 
are not affected by a measurement performed earlier on the other. We have 
seen that analogous no-signaling conditions are satisfied for usual classical 
and quantum bipartite systems and that classical information can be 
superluminally transmitted with certainty between independent systems when 
the condition is violated. In all considered cases, the proposed requirement 
is equivalent to the convex-linearity of the probability measure transformations 
describing finite-time evolutions.

The implications of the dynamical no-signaling condition for classical-quantum 
interactions depend on the hybrid approach used, defined by additional 
assumptions. For hybrid dynamics with classical trajectories, i.e., with only 
Dirac probability measures for classical systems, it implies that the classical 
degrees of freedom can influence the quantum ones, but the latter cannot 
react on the former. The absence of quantum reaction can also be shown 
when classical localized states can only be described approximately, as, e.g., 
for hybrid approaches with classical probability density functions. For 
hybrid approaches limited to pure quantum states, the classical and the 
quantum degrees of freedom are also necessarily uncorrelated. In this case, 
it is the absence of a classical reaction that follows from the dynamical 
no-signaling condition. 

For descriptions with all hybrid probability measures and those with classical 
probability density functions, the following results have been obtained. First, 
the dynamical no-signaling requirement is fulfilled with a hybrid ancillary system 
as soon as it holds with a classical or a quantum ancilla. For no-signaling 
hybrid dynamics such that pure states of quantum systems remain pure, there 
is no classical reaction when the classical and the quantum degrees of freedom 
are initially uncorrelated, whether or not the latter are in a pure state. For 
no-signaling dynamics that do not generate classical-quantum correlations, 
the hybrid interactions can be described as follows. The probability measure 
of the classical subsystem or the state of the quantum one does not depend on 
the initial state of the quantum subsystem, and similarly with the initial 
probability measure of the classical one.
\begin{appendix}
\section*{Appendix: Quantifying classical-quantum correlations}
Basic requirements for a measure of classical-quantum correlations are that it 
must vanish for uncorrelated hybrid states and that it cannot increase under 
convex-linear hybrid transformations describing independent classical and quantum 
systems \cite{PRA1}. We restrict ourselves to the case of a hybrid system $h$ 
with a discrete classical subsystem $c_h$, i.e., the sample space $X$ is a subset 
of $\mathbb{N}$ and $\cal A$ is the set of all subsets of $X$. In this case, the 
counting measure is the natural reference measure and a hybrid probability 
measure $w$ determines a unique hybrid state $\omega$ such that 
$w : (A,E) \mapsto \sum_{x \in A} \mathrm{tr}(\omega(x)E)$, see 
Sec.~\ref{Hscrm}. As a measure of correlations between $c_h$ and $q_h$, we 
propose the classical-quantum mutual information 
$$I(\omega)= \sum_{x \in X} \mathrm{tr} (\omega_x \log \omega_x) 
- \sum_{x \in X} f_x \log f_x - \mathrm{tr} (\rho \log \rho) , $$  
where the argument $x$ is written as an index, 
$f : x \mapsto \mathrm{tr} \omega_x$ is the probability density function of $c_h$, 
and $\rho=\sum_{x \in X} \omega_x$ is the state of $q_h$ \cite{PRA}. It clearly 
vanishes for any uncorrelated hybrid state $\omega : x \mapsto f_x \rho$ and is 
equal to the Holevo quantity of the ensemble 
$\{ (f_x, \omega_x/f_x) \}_{x|f_x>0}$ \cite{Ho} and to the quantum mutual 
information 
$S(\hat \omega \Vert \rho \otimes \sum_{x \in X}  f_x |x \rangle \langle x |)$ 
of the effective quantum state $\hat \omega = 
\sum_{x \in X} \omega(x) \otimes |x \rangle \langle x |$ on 
${\cal H} \otimes {\cal H}'$, where $S$ is the quantum relative entropy, 
${\cal H}'$ is a Hilbert space with a dimension equal to the cardinality of $X$, and 
$| x \rangle$ denote the elements of an orthonormal basis of ${\cal H}'$. As $S$ 
does not increase under any positive trace-preserving linear map from the set of 
trace-class operators on ${\cal H} \otimes {\cal H}'$ to itself applied to both its 
arguments, provided they are positive \cite{MHR}, the measure $I$ meets the 
above-mentioned monotonicity requirement, see the proof below. It follows from 
its definition that $I(\omega)$ cannot exceed the von Neumann entropy of 
$\rho$.
\begin{proof}
Consider any hybrid system $h$ with a discrete classical subsystem $c_h$ and 
any convex-linear probability measure transformation $\cal T$ of $h$. To any 
such transformation corresponds a convex-linear map from the set of states of 
$h$ to itself. This map can be extended to a bounded linear map $T$ from the 
set $\Omega$ of all maps $\omega$ from $X$ to the set of trace-class operators 
on $\cal H$ such that 
$\Vert \omega \Vert_{\Sigma}=\sum_{x \in X} \Vert \omega(x) \Vert$ is finite, 
where $\Vert \cdot \Vert$ denotes the trace norm, equipped with the norm 
$\Vert \cdot \Vert_{\Sigma}$, to itself. In the sums below, $x$, or $x'$, runs over 
$X$ when no set is specified. By construction, if $\omega(x)$ is positive for any 
$x$, then $T(\omega)(x)$ is positive for any $x$ and 
$\sum_{x} \mathrm{tr} T(\omega)(x)
= \sum_{x} \mathrm{tr} \omega(x)$ \cite{PRA}. 

Let ${\cal H}'$ be a Hilbert space with a dimension equal to the cardinality of $X$, 
and $\{ | x \rangle \}_{x \in X}$ be an orthonormal basis of ${\cal H}'$. We build 
a linear map $\hat T$ from the set of trace-class operators on 
${\cal H} \otimes {\cal H}'$, which is a Banach space with the trace norm \cite{BS}, 
to itself as follows. Let $M$ be any such trace-class operator and 
$\omega_M : x \mapsto \mathrm{tr}_{{\cal H}'} 
( M I \otimes |x \rangle \langle x |)$ where $\mathrm{tr}_{{\cal H}'}$ denotes the 
partial trace with respect to ${\cal H}'$. For any $x \in X$, $\omega_M(x)$ is 
trace-class by construction and is positive when $M$ is, since 
$\langle \psi | \omega_M(x) |\psi \rangle
=\langle \psi |\langle x | M |\psi \rangle|x \rangle$ for any 
$|\psi \rangle \in {\cal H}$. Any trace-class operator $M$ can be written as 
$M=P_1-P_2+iP_3-iP_4$ where the operators $P_r$ are positive with finite trace. 
They fulfill $P_1-P_2=(M+M^{\dag})/2$, $P_3-P_4=(M-M^{\dag})/2i$, and 
$\Vert P_r-P_{r+1} \Vert=\mathrm{tr}(P_r+P_{r+1})$, where $r=1$ or $3$ 
\cite{BS} and so, as $\omega_{P_r}(x)$ is positive for any $x$, see above, 
$\Vert \omega_M \Vert_{\Sigma} \le 
\sum_{x,r} \mathrm{tr} (P_r I \otimes |x \rangle \langle x |) 
\le 2 \Vert M \Vert$, which implies that $\omega_M \in \Omega$. Let
$M_n=\sum_{x \in X_n} T(\omega_M)(x) \otimes |x \rangle \langle x |$ and 
$I_n=\sum_{x \in X_n} \Vert T(\omega_M)(x) \Vert$ where 
$X_n=\{ x \in X | x \le n \}$, which are such that 
$\lim_{n \rightarrow \infty} I_n=\Vert T(\omega_M) \Vert_{\Sigma}$ and 
$M_n$ is trace-class, as $\Vert M_n \Vert \le I_n$. Since 
$\Vert M_m-M_n \Vert \le |I_m-I_n|$ for any integers $m$ and $n$, 
$(M_n)_n$ is a Cauchy sequence and we define $\hat T : M \mapsto 
\sum_{x} T(\omega_M)(x) \otimes |x \rangle \langle x |$, where the sum 
converges in trace norm.

Consider any positive trace-class operator $M$ on ${\cal H} \otimes {\cal H}'$. 
We have seen above that $\omega_M(x)$ is positive for any $x$ and so that 
$T(\omega_M)(x)$ is positive for any $x$, which implies that $\hat T (M)$ is 
positive. Moreover, it follows from the properties of $T$ and the definition of 
$\omega_M$ that $\mathrm{tr} \hat T(M)
=\sum_x \mathrm{tr} T(\omega_M)(x)=\mathrm{tr} M$. In conclusion, $\hat T$ 
is a positive trace-preserving linear map, and so 
$S(M \Vert N) \ge S(\hat T(M) \Vert \hat T(N))$ for any positive trace-class 
operators $M$ and $N$ \cite{MHR}. In particular, one has, for any state 
$\omega$ of $h$, 
$S(\hat T (\hat \omega) \Vert \hat T(\rho \otimes \rho')) \le I(\omega)$, 
where $\hat \omega =\sum_{x} \omega(x) \otimes |x \rangle \langle x |$, 
$\rho=\sum_{x} \omega(x)$, and 
$\rho'=\sum_{x} \mathrm{tr} (\omega(x)) |x \rangle \langle x |$. It remains to 
show that the left-hand side of this inequality is equal to $I(T(\omega))$ when 
$\cal T$ describes independent classical and quantum subsystems. 

It results from the definition of $\hat T$ that 
$\hat T (\hat \omega)=\sum_{x} T(\omega)(x) \otimes |x \rangle \langle x |$ 
and $\hat T (\rho \otimes \rho')
=\sum_{x} T(\rho f)(x) \otimes |x \rangle \langle x |$, with 
$f : x \mapsto \mathrm{tr} \omega(x)$. Since $T$ describes independent 
subsystems, this last state can be rewritten as $\hat T (\rho \otimes \rho')
=T_q(\rho) \otimes \sum_{x} T_c(f)(x) |x \rangle \langle x |$, where $T_q$ is a 
map from the set of states of $q_h$ to itself and $T_c$ is a map from the set of 
probability density functions of $c_h$ to itself. As $c_h$ is discrete, any hybrid 
state $\omega$ can be written as $\omega=\sum_{x'} f(x') \eta(x') \epsilon_{x'}$, 
where the sum converges in norm $\Vert \cdot \Vert_{\Sigma}$, 
$\eta(x')=\omega(x')/f(x')$ when $f(x')>0$ and is equal to any given density 
operator otherwise, and $\epsilon_{x'} : x \mapsto \delta_{x'}(\{ x \})$. As $T$ is 
a bounded linear map, one has 
$T(\omega)=\sum_{x} f(x) T_q(\eta(x)) T_c(\epsilon_{x})$ where the sum 
converges in norm $\Vert \cdot \Vert_{\Sigma}$ and so 
$\mathrm{tr} (T(\omega)(x))=\sum_{x'} f(x') T_c(\epsilon_{x'})(x)$ for any $x$. 
For any $| \psi \rangle \in {\cal H}$, one has 
$\langle \psi |\sum_x T(\omega) (x) | \psi \rangle
=\langle \psi | \sum_{x} f(x) T_q(\eta(x))| \psi \rangle$ as the two sums converge 
in trace norm and sums with non-negative terms can be interchanged. The map 
$T_q$ can be written as $T_q : \rho \mapsto \sum_x T(\rho g)(x)$, where $g$ is 
any positive map on $X$ such that $\sum_x g(x)=1$. It can be extended to the set 
of trace-class operators on $\cal H$. As $T$ is a bounded linear map with respect 
to the norm $\Vert \cdot \Vert_{\Sigma}$, $T_q$ is a bounded linear map with 
respect to the trace norm, and so 
$T_q(\rho)=\sum_{x} f(x) T_q(\eta(x))=\sum_x T(\omega) (x)$. The map $T_c$ is 
given by $T_c(f) : x \mapsto \mathrm{tr} T(\tau f)(x)$, where $\tau$ is any density 
operator on $\cal H$. It can be extended to a bounded linear map on the set of all 
complex maps $z$ on $X$ with finite $| z |_{\Sigma}=\sum_{x} | z(x) |$ equipped 
with the norm $| \cdot |_{\Sigma}$ and so 
$T_c(f)(x)=\sum_{x'} f(x') T_c(\epsilon_{x'})(x)
=\mathrm{tr} (T(\omega)(x))$ for any $x$, which completes the proof that 
$S(\hat T (\hat \omega) \Vert \hat T(\rho \otimes \rho')) 
= I(T(\omega))$, and hence $I(T(\omega)) \le I(\omega)$.
\end{proof}
\end{appendix}


\begin{thebibliography}{99}
\bibitem{fn} This can be shown by writing, for $x \notin A$, 
$w_n(A,E) \le p_n(A)$ and, for $x \in A$, 
$|w_n(A,E)-p_n(A)\mathrm{tr}(\eta(x)E)|
\le \int_{A \setminus \{ x \}} |\mathrm{tr}((\eta(x')-\eta(x))E)| dp_n(x') 
\le 2 p_n(A \setminus \{ x \})$.
\end{thebibliography}
\end{document}